\documentclass[11pt]{article}

\usepackage{latexsym,mathrsfs}
\usepackage{amsmath,amssymb} 
\usepackage{amsthm,enumerate,verbatim}
\usepackage{amsfonts}
\usepackage{graphicx}
\usepackage{algorithm}
\usepackage{algorithmic}
\usepackage{url}
\usepackage{graphicx}
\usepackage{subcaption}
\usepackage{mathtools}
\usepackage[toc,page]{appendix}

\renewcommand{\algorithmicrequire}{\textbf{Input:}}
\renewcommand{\algorithmicensure}{\textbf{Output:}}

\newtheorem{lemma}{Lemma}

\newtheorem*{definition*}{Definition}

\newtheorem{assumption}{Assumption}

\DeclareMathOperator{\argmax}{argmax} 
\DeclareMathOperator{\argmin}{argmin}

\title{ Successive Projection Algorithm Robust to Outliers } 
\date{}

\author{Nicolas Gillis\thanks{Email: nicolas.gillis@umons.ac.be. This work was supported by the European Research Council (ERC starting grant no 679515), and the Fonds de la Recherche Scientifique - FNRS and the Fonds Wetenschappelijk Onderzoek - Vlanderen (FWO) under EOS Project no O005318F-RG47.  } \\ 
Department of Mathematics and Operational Research \\ 
Facult\'e Polytechnique, Universit\'e de Mons \\ 
Rue de Houdain 9, 7000 Mons, Belgium   
	}

\begin{document}

\maketitle

\begin{abstract}
The successive projection algorithm (SPA) is a fast algorithm to tackle separable nonnegative matrix factorization (NMF). Given a nonnegative data matrix $X$, SPA identifies an index set $\mathcal{K}$ such that there exists a nonnegative matrix $H$ with $X \approx X(:,\mathcal{K})H$. SPA has been successfully used as a pure-pixel search algorithm in hyperspectral unmixing and for anchor word selection in document classification. Moreover, SPA is provably robust in low-noise settings. The main drawbacks of SPA are that it is not robust to outliers and does not take the data fitting term into account when selecting the indices in $\mathcal{K}$. In this paper, we propose a new SPA variant, dubbed Robust SPA (RSPA), that is robust to outliers while still being provably robust in low-noise settings, and that takes into account the reconstruction error for selecting the indices in $\mathcal{K}$. We illustrate the effectiveness of RSPA on synthetic data sets and hyperspectral images.   
\end{abstract}


\section{Introduction} 

Given a nonnegtive data matrix $X \in \mathbb{R}^{m \times n}_+$ and a factorization rank $r$, NMF looks for nonnegative matrices 
$W \in \mathbb{R}^{m \times r}_+$ and 
$H \in \mathbb{R}^{r \times n}_+$
such that $WH \approx X$. NMF can be used for example in image analysis, document classification and hyperspectral unmixing; see, e.g.,~\cite{cichocki2009nonnegative, xiao2019uniq} and the references therein. 
However, NMF is NP-hard in general~\cite{vavasis2009complexity}. 
Separable NMF is an NMF variant where it is assumed that $W = X(:,\mathcal{K})$ for some index set $\mathcal{K}$ of size $r$. This means that the basis matrix $W$ is contained within the data set.  
This assumption makes sense in several applications including document classification~\cite{arora2013practical} and hyperspectral unmixing~\cite{ma2014signal}. 
In hyperspectral unmixing, the separability assumption is known as the pure-pixel assumption and requires that for each material present in the image (also called endmember), at least one pixel contains only that material. 
The pure-pixel assumption has been used for a long time in the literature~\cite{ma2014signal}, but it is only rather recently that separable NMF algorithms with provable guarantees have been proposed~\cite{AGKM11}. 
Among these algorithms, the successive projection algorithm (SPA) is one of the most popular ones: it is very fast, simple to implement and robust to noise~\cite{gillis2014fast}. 
It was first introduced in~\cite{MC01} and has been rediscovered many times; see the discussion in~\cite{gillis2014and}. 
However, SPA has two main drawbacks: 
(1)~SPA is very sensitive to outliers, 
and (2)~SPA does not take directly the data fitting term into account to select the indices in $\mathcal{K}$. 
 SPA can be made robust to outliers either by properly preprocessing the data set and removing the outliers, or using a proper post-processing of the index set $\mathcal{K}$~\cite{gillis2014fast}. However, these approaches do not alleviate the second drawback of SPA. Moreover, it would be useful to have an SPA variant robust to outliers in case these pre- and/or post-processings fail to identify all outliers. 

In this paper, we propose a new variant of SPA that is robust to outliers and takes directly the data fitting term into account to select the indices in $\mathcal{K}$. Moreover, this variant retains the good properties of SPA: it is fast, simple to implement and robust to noise. 
The paper is organized as follows. 
In Section~\ref{sec:spa}, we recall how SPA works and its properties. 
In Section~\ref{sec:rspa}, we present our new SPA variant, dubbed robust SPA (RSPA), that is robust to outliers and takes the data fitting term into account in the selection step. 
In Section~\ref{sec:numexp}, we illustrate the effectiveness of this new approach on synthetic data sets and hyperspectral images.

\section{The successive projection algorithm} \label{sec:spa} 

Algorithm~\ref{alg:spa} gives the pseudocode of SPA. 
SPA sequentially identifies indices in $\mathcal{K}$ using two steps: at iteration $k$, 
given the current residual matrix $R$, 

\begin{itemize}

\item Selection step: 
add the index $k$ to $\mathcal{K}$ that maximizes \mbox{$f(R(:,k))$} where $f$ is a given function (see Assumption~\ref{fass1}). 

\item  Projection step: project the residual $R$ onto the orthogonal complement of $R(:,k)$ (step~5 of Algorithm~\ref{alg:spa}). 

\end{itemize}

\algsetup{indent=2em}
\begin{algorithm}[ht!]
\caption{SPA \label{alg:spa}}
\begin{algorithmic}[1] 
\REQUIRE Nearly separable matrix $X$ (Assumption~\ref{asssep}),  
the number $r$ of columns to be extracted, and a strongly convex function $f$ (Assumption~\ref{fass1}). 

\ENSURE Index set $\mathcal{K}$ such that 
$X\approx X(:,\mathcal{K})H$ with \mbox{$H \geq 0$}.
    \medskip 
\STATE Let $R = X$, 
$\mathcal{K} = \{\}$, $k=1$.  
\WHILE {$R \neq 0$ and $k \leq r$}   
\STATE $k^* = \argmax_k f(R(:,k))$. 
\STATE $u_j = R(:,k^*) / ||R(:,k^*)||_2$.  \vspace{0.1cm} 
\STATE $R \leftarrow (I-{u_j u_j^T})R = R - u_j (u_j^TR)$. 
\STATE $\mathcal{K} = \mathcal{K} \cup \{k^*\}$, $k = k+1$.
\ENDWHILE
\end{algorithmic}
\end{algorithm}

Let us define the class of matrices for which SPA will provably identify a subset $\mathcal{K}$ such that there exists a nonnegative matrix $H$ with $X = X(:,\mathcal{K})H$. 
\begin{assumption} \label{asssep}
The \emph{separable matrix} $X \in \mathbb{R}^{m \times n}$ can be written as 
 $X = W H = W [I_r, H'] \Pi$,  
where $W \in \mathbb{R}^{m \times r}$ has rank $r$,
$H \in \mathbb{R}^{r \times n}_+$, 
$I_r$ is the identity matrix of size $r$, 
$\Pi$ is a permutation matrix, 
and the sum of the entries of each column of $H'$ is at most one. 
\end{assumption} 

\newpage 

Let us make a few remarks:  

\begin{itemize} 
\item SPA is similar to vertex component analysis (VCA)~\cite{nascimento2005vertex}: 
the main difference is in the selection step where VCA uses a linear function, which is not robust to noise. 

\item If the sum-to-one-constraint on the columns of $H'$ is not satisfied by the input matrix $X$, it can be obtained by normalizing each column of $X$ to have uni $\ell_1$ norm~\cite{gillis2014fast}. 

\item In the absence of noise, separable NMF is equivalent to identifying the vertices of a set of points; see~\cite{ma2014signal, xiao2019uniq}. 

\end{itemize}

In the presence of bounded noise, SPA will be able to identify $\mathcal{K}$ such that $X(:,\mathcal{K}) \approx W$ (up to permutation); see~\cite{gillis2014fast} where error bounds are provided. 
For this result to hold, the function $f$ must satisfy the following assumption. 
\begin{assumption} \label{fass1} 
The function $f:\mathbb{R}^m \to \mathbb{R}_+$ is strongly convex with parameter $\mu_f > 0$, its gradient is Lipschitz continuous with constant $L_f$, and its global minimizer is the all-zero vector with $f(0) = 0$. 
\end{assumption} 

The standard variant of SPA uses $f(x) = ||x||_2^2$, and is the most robust to noise according to the analysis in~\cite{gillis2014fast} since the error bound depends on the ratio $L_f/\mu_f$. 
This ratio is the conditioning of the function $f$ and denoted $\kappa_f = L_f/\mu_f \geq 1$.  
For $f(x) = ||x||_2^2$, we have $\kappa_f = 1$. 
However, SPA remains robust in low-noise settings as long as Assumption~\ref{fass1} is satisfied. 
For example, one can choose any quadratic function $f = \frac{1}{2} x^T Q x$ where $Q$ is positive definite, and we have 
$\kappa_f = \frac{\lambda_{\max}(Q)}{\lambda_{\min}(Q)}$.  
Moreover, since the analysis of SPA is sequential, the analysis still holds if one chooses different functions $f$ to select the index to put in $\mathcal{K}$ at each step of SPA, as long as they satisfy Assumption~\ref{fass1}.

\section{Robust SPA} \label{sec:rspa} 

The main contribution of this paper is to leverage the flexibility of SPA in choosing the function $f$ 
in order to make SPA robust to outliers by taking into account the residual error during the selection step. 
Algorithm~\ref{alg:rspa}, which we refer to as robust SPA (RSPA), is our proposed robust variant of SPA. It only differs from SPA in \emph{the selection step}. 
\algsetup{indent=2em}
\begin{algorithm}[ht!]
\caption{Robust SPA \label{alg:rspa}}
\begin{algorithmic}[1] 
\REQUIRE 
Nearly separable matrix $X$ (Assumption~\ref{asssep}), 
the number $r$ of columns to be extracted, 
number of candidates $d \geq 1$, 
error norm parameter $p > 0$, 
diversification parameter $\beta > 1$. 

\ENSURE 
Index set $\mathcal{K}$ such that 
$X \approx$ \mbox{$X(:,\mathcal{K})$}$H$ with $H \geq 0$.
    \medskip 

\STATE Apply SPA to the input matrix $X$ to extract $r$ indices, but replace step 3 of SPA by the following step:   

\STATE Pick $k^*$ using Algorithm~\ref{alg:divers} with input $(R,d,p,\beta)$. 





\end{algorithmic}
\end{algorithm}  

Let us explain the selection step of RSPA described in Algorithm~\ref{alg:divers}. 
\algsetup{indent=2em}
\begin{algorithm}[ht!]
\caption{Selection step for RSPA \label{alg:divers}}
\begin{algorithmic}[1] 
\REQUIRE 
Input matrix $R$,  
diversification parameter $d \geq 1$, 
error norm parameter $p > 0$, 
diversification parameter $\beta > 1$. 

\ENSURE Index $k^*$. 

    \medskip 
\STATE Let $Y = R$, 
$P_1 = I$. 

\FOR {$i = 1$ : $d$}  
 
\STATE $k(i) = \argmax_j ||Y(:,j)||_2$. 

\STATE $u_i = R(:,k(i)) / ||R(:,k(i))||_2$.  

\STATE $R_i \leftarrow ( I -  u_i u_i^T ) R$.


\STATE $e(i) = \sum_{k} ||R_i(:,k)||_2^p$. 

\STATE $k'(i) = \argmax_k ||R_i(:,k)||_2$.

\STATE Compute $\alpha_i$ as given by \eqref{eq:alpha} (see Lemma~\ref{lem:alpha}) 

with $x = Y(:,k(i))$ and $y = Y(:,k'(i))$. 



\STATE $Y \leftarrow P_{i+1} Y = ( I - \alpha_i u_i u_i^T ) Y$.

\ENDFOR 

\STATE $k^* = k(i^*)$ where $i^* = \argmin_{1 \leq i \leq d} e(i)$. 

\end{algorithmic}
\end{algorithm}  
As opposed to SPA that simply picks the column of the current residual $R$ that maximizes $f$ (step 3 of SPA), RSPA uses $d$ well-chosen quadratic functions $f_i(x)$ $(1 \leq i \leq d)$. Each function $f_i$ will correspond to a candidate column of $R$ with index $k(i) = \argmax_j f_i(R(:,j))$. Among these candidates, RSPA will select the one with the smallest residual after projection onto its orthogonal complement. 
To measure the norm of the residual, 
we use the $\ell_p$ norm of the vector containing the $\ell_2$ norms of the columns of the residual, but many other measures could be used. 
We have observed that using $p=1$ works well, as it is less sensitive to large entries; see Section~\ref{sec:numexp}. 
As long as the functions $\{f_i\}_{i=1}^d$ satisfy Assumption~\ref{fass1}, \emph{RSPA is guaranteed to be robust in low-noise settings}. Moreover, this selection step will be more robust to outliers because outliers lead in general to a smaller decrease in the residual since they are less/not correlated with the data points.

It remains to explain how the functions $\{f_i\}_{i=1}^d$ are generated. 
Note that one could be tempted to generate them randomly but this will most likely still lead to the identification of outliers like in SPA; 
in particular if an outlier has a very large norm. 
For example, if one uses quadratic functions  
$x^T P^T P x = ||Px||_2^2$ where the entries of $P$ are randomly generated, an outlier with a large norm will most likely also have a large value for $||Px||_2^2$.  
Hence we generate $\{f_i\}_{i=1}^d$ so that the associated data points maximizing them are well spread in the data set. To do so, we define 
\[
f_i(x)  
= ||P_i P_{i-1} \dots P_1 x||_2^2 
\; \text{ for } \; 
1 \leq i \leq d, 
\]
where $P_1 = I_m$ and 
$P_{i+1} = I_m - \alpha_i u_i u_i^T$ for some well chosen $\alpha_i \in (0,1)$ and $u_i$ with unit $\ell_2$ norm.   
Since $\alpha_i \in (0,1)$ and $||u_i||_2 = 1$, the matrices $\{P_{i}\}_{i=1}^d$ are positive definite (all eigenvalues are equal to one except for one which is equal to $1-\alpha_i$) hence $\{f_{i}\}_{i=1}^d$ are strongly convex functions. Note that $f_1(x) = ||x||_2^2$ hence, for $d = 1$, RSPA is equivalent to SPA with $f(x) = ||x||_2^2$.  
The matrices $\{P_{i}\}_{i=1}^d$ are chosen such that there is a \emph{diversification} of the data points maximizing the functions $\{f_i\}_{i=1}^d$.   
Our strategy is described in Algorithm~\ref{alg:divers} and works sequentially as follows: 
Given an input matrix $R$ (the residual after $k$ steps of RSPA), for $i=1,2,\dots,d$:  

\noindent \textbf{step 3}. Identify the index $k(i)$ that correspond to the data point $R(:,k(i))$ that maximizes $f_i(x)$. 

\noindent \textbf{steps 4-6}. Let $u_i = R(:,k(i)) / ||R(:,k(i))||_2$, and let $R_i = (I - u_iu_i^T)R$ be the projection of $R$ onto the orthogonal complement of $u_i$. 
The $\ell_p$ norm of the $\ell_2$ norms of the columns of $R_i$, denoted $e(i)$, will allow us to select the index among $\{k(i)\}_{i=1}^d$ such that this norm is minimized (step 11). 

\noindent \textbf{steps 7-9}. Identify the index $k'(i)$ corresponding to the column of $R_i$ with largest norm. 
We choose $\alpha_i$ such that 
\[
f_{i+1}(R(:,k'(i))) = \beta f_{i+1}(R(:,k(i)))  
\text{ with } 
\beta > 1. 
\] The value of $\alpha_i$ that achieves this is given in Lemma~\ref{lem:alpha} (see below). 
This guarantees that, at the next step, $f_{i+1}$ will not identify $k(i)$ again since there is at least one data point with value $\beta$ times larger for $f_{i+1}$. 
Note that we simplify the computation of $f_{i+1}(R(:,k))$ by  introducing the matrix $Y$ initialized as $Y=R$ and updated at each step as $Y \leftarrow P_{i+1}Y$ (step 9) so that $f_{i+1}(R(:,k)) = ||Y(:,k)||_2^2$ for all~$k$.

\begin{lemma} \label{lem:alpha} Let $x$ and $y$ be two non-zero vectors not multiple of one another, $u = x/||x||_2$ and $||x||_2 > ||y||_2$. For $\beta > 1$, 
\begin{equation} \label{eq:alpha}
\alpha^* =
1 - \sqrt{1 - \frac{\beta ||x||_2^2 - ||y||_2^2}{\beta (u^Tx)^2 - (u^T y)^2}}. 
\end{equation} 
is the unique solution for $\alpha \in (0,1)$ of 
\[
||(I- \alpha uu^T) y||_2^2 
= 
\beta  
||(I- \alpha uu^T) x||_2^2. 
\] 
\end{lemma}
\begin{proof}
Using $||(I- \alpha uu^T) z||_2^2
=
||z||_2^2 - 
2 \alpha (u^T z)^2 
+ \alpha^2 (u^T z)^2$,   
the solution of the above problem is a root of 
$\Delta \alpha^2
- 2  \Delta \alpha
+ 
\beta ||x||_2^2 - ||y||_2^2$ 
where $\Delta = \beta (u^Tx)^2 - (u^T y)^2$. Note that $u^Tx = ||x||_2 > u^T y$ since $u^T y \leq ||y||_2 < ||x||_2$ hence $\Delta > \beta ||x||_2^2 - ||y||_2^2$. 
We obtain 
$\alpha^*   
= 
\frac{ \Delta - \sqrt{\Delta^2 - \Delta (\beta ||x||_2^2 - ||y||_2^2) } }{ \Delta } \in (0,1)$. 
\end{proof}

\paragraph{Computational cost} It can be checked that SPA runs in $O(mnr)$ operations~\cite{gillis2014fast}, 
while RSPA runs in $O(mnrd)$. The main computational cost lies in the selection and projection steps, each in $O(mn)$ operations. 

\section{Numerical experiments} \label{sec:numexp}

We write RSPA($d$,$p$,$\beta$) to refer to RSPA with parameters ($d$,$p$,$\beta$). 
The code is available from 
\begin{center}
\url{https://sites.google.com/site/nicolasgillis/code}. 
\end{center}


\subsection{Synthetic data sets} 

Let $r=10$, $n=1000$ and  the value of $m$ is varied from 10 to 50. 
Each entry of $W \in \mathbb{R}^{m \times r}$ is generated randomly using the uniform distribution $\mathcal{U}(0,1)$ in the interval $[0,1]$. Each entry of $H \in \mathbb{R}^{r \times n-r}$ is generated in the same way, but then each column of $H$ is normalized so that Assumption~\ref{asssep} holds: $H(:,j) \leftarrow H(:,j)/||H(:,j)||_1$. Finally, we take 
$X = W [I_r, H]$ to which we add 10 outliers whose entries are generated randomly using the normal distribution $\mathcal{N}(0,1)$ of mean zero and variance 1. 
For each value of $m$, we generate 100 such matrices. 
Figure~\ref{synt} reports the percentage of correctly identified columns of $W$ by SPA and by RSPA with various combinations of the parameters, with $d\in\{10,20,40,80\}$, $p\in\{1,2\}$ and $\beta \in \{2,4,8\}$.  
\begin{figure}[ht!] 
\centering 
   \includegraphics[width=0.75\textwidth]{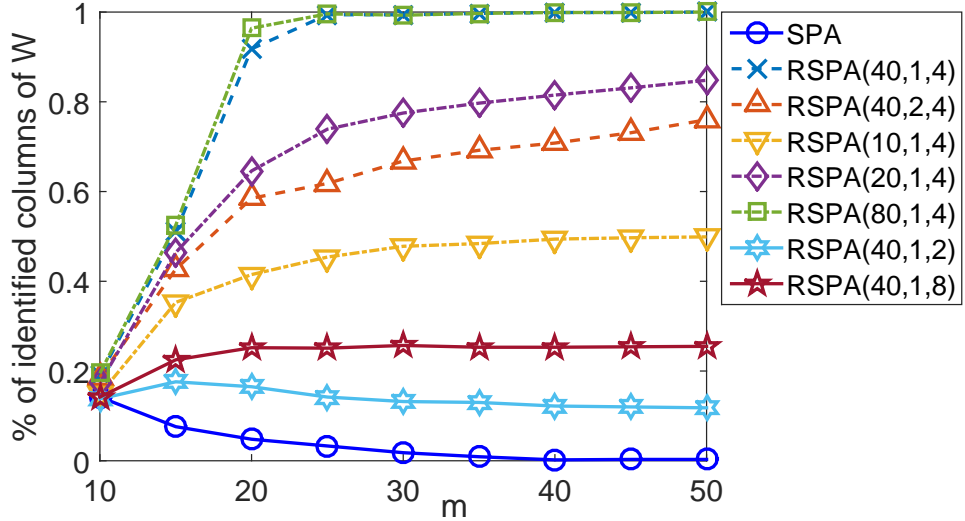} 
\caption{Percentage of recovered columns of $W$ for the synthetic data sets using SPA and RSPA.}
\label{synt}
\end{figure}

We observe the following: 

\begin{itemize}

\item When $m$ is small, no algorithm is able to recover all columns of $W$. The reason is that the outliers and the columns of $W$ are less separated (they are linearly dependant for example when $m=r=10$). 

\item SPA performs very poorly. 
The reason is that the norm of the outliers is larger than that of the columns of $W$; in fact, $\mathbb{E}(x^2) = 1/3$ for $x \sim \mathcal{U}(0,1)$ while $\mathbb{E}(x^2) = 1$ for $x \sim \mathcal{N}(0,1)$. 
Note that any other greedy algorithm such as VCA~\cite{nascimento2005vertex} or the successive nonnegative projection algorithm (SNPA)~\cite{gillis2014successive} would fail as well.  

\item RSPA does not perform well when $p=2$ as it is more sensitive to large entries in the residual hence to outliers. We have also tried $p=0.5$ and it performed similarly as $p=1$. 

\item RSPA performs best for $\beta = 4$. 
The parameter $\beta$ should not be chosen too large as it  makes $P_i$ ill conditioned since $\alpha_i$ will be close to 1, nor too small as it does not provide a good diversification. 

\item RSPA does not perform well when $d$ is too small ($d \leq 20$): in that case, RSPA is not able to avoid outliers. 

\end{itemize}

To summarize, RSPA performs best when $d \geq 40$, $p=1$ and $\beta = 4$: more than 99\% of the columns of $W$ are correctly identified for $m \geq 25$. 
Ideally, $d$ should be of the order of the number of outliers. In fact, if $d$ is smaller than the number of outliers, the diversification procedure could only identify outliers hence fail. This is particularly crucial when the outliers have a larger norm than the inliers (see also the experiment on the San Diego hyperspectral image). 
However, choosing $d$ too large makes RSPA slower as it runs in $O(mnrd)$ operations.

\subsection{Hyperspectral images} 

Let us compare SPA and RSPA on three widely used hyperspectral images (HSIs) that are described for example in~\cite{gillis2014hierarchical}:  

\begin{itemize}

\item Urban: 162 spectral bands, 
$307 \times 307$ pixels and 6 endmembers. 
It contains a few outliers that correspond to materials present in small proportions. 

\item San Diego: 158 spectral bands, 
$400 \times 400$ pixels and 8 endmembers. It contains quite a few outliers (see Figure~\ref{hsi_sd}). 

\item Cuprite: 188 spectral bands, 
$250 \times 191$ pixels and 15 endmembers. 
It does not contain outliers. 

\end{itemize}

Table~\ref{tab:hsi} reports the relative approximation error 
\begin{equation} \label{eq:re}
\min_{H \geq 0} \frac{||X - X(:,\mathcal{K})H||_F}{||X||_F}
\end{equation}
for the index sets $\mathcal{K}$ extracted by SPA,  RSPA(10,1,4) and RSPA(20,1,4). 
\begin{table}[ht!]
\caption{
Relative error~\eqref{eq:re} and, in brackets, computational time in seconds of SPA, RSPA(10,1,4) and RSPA(20,1,4) applied on three hyperspectral images. \label{tab:hsi}
} 
 \begin{center}
 \begin{tabular}{|c|c|c|c|}
 \hline 
    &      Urban  &  San Diego &  Cuprite \\   \hline  
SPA          
& 9.58   (1.4) &  12.62 \; (2.9)   &  1.83 (1.9) \\   
RSPA(10,1,4) 
& 7.65 (31) &   \; 6.63 \; (64)   &   1.78 \ (52) \\  
RSPA(20,1,4) 
& 6.66 (59) &   \; 6.03  (124)   &   1.83 \ (83)
 \\ \hline 
\end{tabular}
\end{center}
\end{table}

\begin{figure*}[ht!] 
\centering 
   \includegraphics[width=\textwidth]{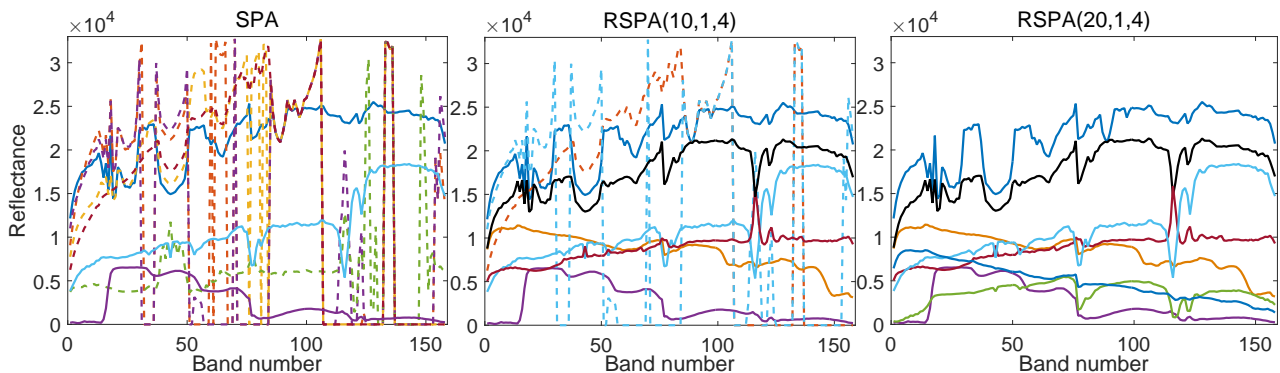} 
\caption{Columns of $X$ extracted by SPA, RSPA(10,1,4) and RSPA(20,1,4) on the San Diego HSI. Dashed lines represent outliers.}
\label{hsi_sd}
\end{figure*}

We observe the following: 

\begin{itemize}

\item For Urban, RSPA variants allow to slightly reduce the relative error compared to SPA. 
The gain is appreciable (from 9.58\% to 7.65\% for RSPA(10,1,4) and to 6.66\% for RSPA(20,1,4)) but not significant because the outliers are endmembers corresponding to materials present in small proportions and sharing  similarities with the main endmembers. 

\item For San Diego, RSPA allows a significant reduction of the relative error; from 12.62\% to 6.63\% for RSPA(10,1,4) and to 6.03\% for RSPA(20,1,4). 
Figure~\ref{hsi_sd} displays the 8 endmembers extracted by each algorithm. We observe on Figure~\ref{hsi_sd} that SPA identifies 5 outliers, 
RSPA(10,1,4) only 2, and RSPA(20,1,4) none. 
This confirms our observations made on synthetic data sets: the parameter $d$ should be chosen properly so as to allow RSPA to avoid extracting outliers. Note however that the relative errors of RSPA(10,1,4) and RSPA(20,1,4) are relatively close because the outliers have a very large norm. 

\item  For Cuprite, SPA and RSPA  provide comparable results because of the absence of outliers. 


\item In terms of computational time, RSPA is between $2d$ to $3d$ times slower than SPA. This is expected since RSPA requires $O(d)$ times more operations than SPA. 

\end{itemize}

\section{Conclusion} 

We have proposed a new variant of SPA, namely Robust SPA (RSPA), which is robust to outliers by taking into account the residual error to identify important columns in the data set, while remaining robust in low-noise settings.  
We have illustrated the effectiveness of RSPA on synthetic data sets and hyperspectral images. 
A similar enhancement could be brought to other greedy separable NMF algorithms, such as VCA and SNPA. 
Further work includes a thorough analysis of the behavior of RSPA under different choices of the parameters in various conditions, as well as a rigorous robustness analysis of RSPA with explicit error bounds depending on the noise level and the number of outliers. 

\small 

\bibliographystyle{spmpsci} 
\bibliography{Biography}

\end{document}